\documentclass[11pt,a4paper]{article}
\usepackage{url}

\usepackage{graphicx}
\usepackage{xspace}
\newcommand{\ig}[2][scale=1.0]{\includegraphics[#1]{#2}}

\newcommand{\etal}{et al.~}

\newcommand{\ie}{i.e.~}

\newcommand{\binom}[2]{{#1 \choose #2}}
\newcommand{\marge}[1]{}

\newcommand{\adj}{\textrm{adj}}
\newcommand{\depth}{{\rm depth}}

\DeclareSymbolFont{AMSb}{U}{msb}{m}{n}
\DeclareSymbolFontAlphabet{\mathbb}{AMSb}

\newtheorem{theorem}{Theorem}[section]

\newtheorem{lemma}[theorem]{Lemma}
\newtheorem{corollary}[theorem]{Corollary}

\newtheorem{definition}[theorem]{Definition}
\newtheorem{example2}[theorem]{Example}

\newenvironment{proof}{\begin{genproof}}{\end{genproof}}

\newenvironment{genproof}[1][]{\begin{trivlist}\item \textbf{Proof#1:} }{\nolinebreak\qquad\nolinebreak\framebox(5,5)[lb]{}\nolinebreak\end{trivlist}}

\newenvironment{genproofnoqed}[1][]{\begin{trivlist}\item \textbf{Proof#1:} }{\nolinebreak\end{trivlist}}

\newcommand{\knip}[1]{}

\newcommand{\burn}{{\sc{Saving All But $k$ Vertices}}\xspace}
\newcommand{\burntwo}{{\sc{Saving All But $k$ Vertices II}}\xspace}
 \newcommand{\maxprotect}{{\sc{Maximum $k$-Vertex Protection}}\xspace}
 \newcommand{\maxprotectdecision}{{\sc{$k$-Vertex Protection}}\xspace}
 \newcommand{\save}{{\sc{Saving $k$ Vertices}}\xspace}
 \newcommand{\firefighter}{{\sc{Firefighter}}\xspace}
\newlength{\atextwidth}
\newcommand{\problemdef}[4]{
  \vspace{1mm}
\noindent\fbox{
  \begin{minipage}{\atextwidth}
  \begin{tabular*}{\textwidth}{@{\extracolsep{\fill}}lr} #1 & {\bf{Parameter:}} #3 \\ \end{tabular*}
  {\bf{Input:}} #2  \\
  {\bf{Question:}} #4
  \end{minipage}
  }
  \vspace{1mm}
}

\begin{document}

%\title{When is the Firefighter Problem Tractable?}
\title{Parameterized Complexity of Firefighting Revisited}
\author{%
Marek Cygan\footnote{Institute of Informatics, University of Warsaw, Poland, \texttt{cygan@mimuw.edu.pl}}
\and Fedor V. Fomin\footnote{Department of Informatics, University of Bergen, Norway, \texttt{fedor.fomin|E.J.van.Leeuwen@ii.uib.no}}
\and Erik Jan van Leeuwen\footnotemark[2]}
\date{}
\maketitle
\vspace{-0.8cm}

\begin{abstract}
\noindent The \firefighter problem is to place firefighters on the vertices of a graph to prevent a fire with known starting point from lighting up the entire graph. In each time step, a firefighter may be permanently placed on an unburned vertex and the fire spreads to its neighborhood in the graph in so far no firefighters are protecting those vertices. The goal is to let as few vertices burn as possible. This problem is known to be NP-complete, even when restricted to bipartite graphs or to trees of maximum degree three. Initial study showed the \firefighter problem to be fixed-parameter tractable on trees in various parameterizations. We complete these results by showing that the problem is in FPT on general graphs when parameterized by the number of burned vertices, but has no polynomial kernel on trees, resolving an open problem. Conversely, we show that the problem is W[1]-hard when parameterized by the number of unburned vertices, even on bipartite graphs. For both parameterizations, we additionally give refined   algorithms on trees, improving on the running times of the known algorithms.
% Finally, we consider the problem under several other parameterizations, including nonstandard parameters.
\end{abstract}

\section{Introduction} \label{sec:introduction}
The \firefighter problem  concerns a deterministic model of fire spreading through a graph via its edges. The problem has recently received considerable attention 
\cite{FinbowHLS00,MacGillivrayW03}. 
In the model, we are given a graph $G$ with a vertex $s\in V(G)$. At time $t=0$, the fire breaks out at $s$ and vertex $s$ starts burning. At each step $t\geq 1$, first the firefighter protects one vertex not yet on fire---this vertex remains permanently protected---and the fire then spreads     
 from  burning vertices to all unprotected neighbors of these vertices. The process stops when the fire cannot spread anymore. The goal is to find a strategy for the firefighter that minimizes the amount of burned vertices, or, equivalently, maximizes the number of saved, i.e.\ not burned,   vertices. 

It is known that the \firefighter problem is NP-hard, even when restricted to bipartite graphs~\cite{MacGillivrayW03} or trees of maximum degree three~\cite{FinbowKMR07}. However, it is polynomial-time solvable on such trees if the root has degree two~\cite{MacGillivrayW03}. We refer to the survey~\cite{FinbowMac09} for an overview of further combinatorial results on the problem. 
%
%
%The problem is known to be NP-hard on trees of maximum degree 3 \cite{FinbowKMR07}.
%  
The study of the problem from the perspective of parameterized complexity was initiated by Cai, Verbin, and Yang~\cite{cai-verbin-yang}.  
They considered the following parameterized versions of the problem and obtained a number of parameterized algorithms on trees.
% (let us note that the problem remains  NP-hard on trees of maximum degree $3$~\cite{FinbowKMR07}).

%\medskip
The first parameterization  considered by Cai et al.\ in~\cite{cai-verbin-yang}  is by the number of saved  vertices. 

\problemdef{\save}{An undirected graph $G$, a vertex $s$, and an integer $k$.}{$k$}{Is there a strategy to save at least $k$ vertices when a fire breaks out at $s$?}

\noindent Cai et al.\ proved that \save  on trees has  a kernel with $O(k^2)$ vertices.  They also 
gave a  randomized algorithm solving \save on trees in time  $O(4^{k}+n)$, which can be derandomized to a $O(n + 2^{O(k)})$-time algorithm.

%\medskip 

The second  parameterization  is by the number of burned  vertices. 

\problemdef{\burn}{An undirected $n$-vertex graph $G$, a vertex $s$, and an integer $k$.}{$k$}{Is there a strategy to save at least $n-k$ vertices when a fire breaks out at $s$?}

\noindent For \burn on trees,
Cai et al.\ gave a randomized algorithm of running time $O(4^{k}n)$, which can be derandomized to a $O(2^{O(k)}n \log n)$-time algorithm. They left as an open problem whether  \burn  has a polynomial kernel on trees. 

%\medskip 
The last parameterization is by the number of protected vertices, i.e. the number of vertices occupied by firefighters. 

\problemdef{\maxprotect}{An undirected graph $G$, a vertex $s$, and an integer $k$.}{$k$}{What is a strategy that saves the maximum number of vertices by protecting  $k$ vertices when a fire breaks out at $s$?}

\noindent For  \maxprotect 
on trees, Cai et al.\ gave a randomized algorithm of running time  $O(k^{O(k)}n)$, which can be derandomized to a $O(k^{O(k)}n \log n)$-time algorithm. % and deterministic    $O(k^{O(k)}n \log{n})$.
They left open whether the problem has a polynomial kernel on trees, and asked whether there is an algorithm solving the problem on trees in time $2^{o(k\log{k})}n^{O(1)}$.  

We will sometimes consider the decision variant of \maxprotect.

\problemdef{\maxprotectdecision}{An undirected graph $G$, a vertex $s$, an integer $k$, and an integer $K$.}{$k$}{Is there a strategy that saves at least $K$ vertices by protecting $k$ vertices when a fire breaks out at $s$?}

\noindent The unparameterized version of this problem is obviously NP-hard on trees of maximum degree three from the hardness of the \firefighter problem.

\medskip\noindent\textbf{Our results}\ \ We resolve several open questions of 
Cai, Verbin, and Yang~\cite{cai-verbin-yang}.  We also refine and extend some of the   results   of \cite{cai-verbin-yang}.
\begin{itemize} %\setlength{\parskip}{-0.05cm}

\item
In Section~\ref{sec:saving}, we give a deterministic algorithm solving \save on trees in time  $O(2^{k}k^3+n)$, improving the running time  $O(4^{k}+n)$ of the randomized algorithm from  \cite{cai-verbin-yang}.
We also observe that on general graphs the problem is W[1]-hard, which was independently observed by Cai (private communication), but is in FPT when parameterized by $k$ and the treewidth of a graph.  Based on that we derive  that \save is in FPT on graphs of bounded local treewidth, including planar graphs, graphs of bounded genus, apex-minor-free graphs, and graphs of bounded maximum vertex degree.

\item
In Section~\ref{sec:burning}, we   provide  deterministic algorithms solving \burn  in time $O(2^{k}n)$ on trees, and in time  $O(3^{k}n)$ on general graphs. The algorithm on trees improves the $O(4^{k}n)$ running time of  the randomized algorithm  from  \cite{cai-verbin-yang}.  
 We also answer the open question of Cai et al.\ by showing that 
 \burn  has no polynomial kernel on trees of maximum vertex degree four. 

\item
For   \maxprotect, we     answer both    open questions of Cai et al.:  We give a deterministic algorithm solving  \maxprotect on trees in time $O(2^k k n)$ in Section~\ref{sec:saving}, and show that the problem has no polynomial kernel on trees in Section~\ref{sec:burning}.  The no-poly-kernel result was independently obtained by Yang~\cite{Yang2009}.  Based on the parameterized algorithm, we also give an exact subexponential-time algorithm, solving the \firefighter problem on an $n$-vertex tree in time $O(2^{\sqrt{2n}} n^{3/2})$, thus improving on the $2^{O(\sqrt{n}\log{n})}$ running time from~\cite{cai-verbin-yang}.
 On general graphs, we show that the {\maxprotect} problem is W[1]-hard, but is in FPT when parameterized by $k$ and the treewidth of a graph. 
\end{itemize}
Recently, and independent of our work, Bazgan, Chopin, and Fellows~\cite{BazganCF2011} proved several of the results mentioned above. This includes the W[1]-hardness of \save, as well as its membership of FPT on graphs of bounded treewidth, and the membership of FPT of \burn on general graphs, as well as it not having a polynomial kernel on trees. In addition, they consider the parameterization by the vertex cover number of a graph, and the extension of the problem to being able to protect $b$ vertices at any time step. However, they do not consider \maxprotect, algorithms on trees, or exact algorithms.

%Due to space limitations, proofs of lemmas marked with the spade symbol ($\spadesuit$) are postponed to Appendix \ref{app}.

\section{Saving and Protecting Vertices} \label{sec:saving}
In this section, we consider the complexity of \save and \maxprotect. These problems are known to be fixed-parameter tractable on trees, but their complexity on general graphs was hitherto unknown. We resolve this open problem by giving a W[1]-hardness result for both problems. At the other end of the spectrum, we extend the boundary where \save and \maxprotect remain fixed-parameter tractable by giving parameterized algorithms on graphs of bounded treewidth. Finally, we improve the algorithms known to exist for trees.

\subsection{W[1]-Hardness on General Graphs}
We show that \save and the decision variant of \maxprotect are W[1]-hard, even on bipartite graphs. We reduce from the \textsc{$k$-Clique} problem, which is well known to be
 W[1]-hard~\cite{DowneyF99}.

\begin{theorem} \label{thm:saving:hardness}
\save is W[1]-hard, even on bipartite graphs.
\end{theorem}
\begin{proof}
Let $(G,k)$ be an instance of \textsc{$k$-Clique}. We can assume that $G$ has at least $k+1$ vertices that are not isolated, or we can easily output a trivial \textsc{Yes}- or \textsc{No}-instance. We construct the following bipartite graph $G'$ (see Figure~\ref{fig:saving:hardness}). For each edge $(u,v) \in E(G)$, we add a vertex $s_{uv}$, and for each vertex $v \in V(G)$, we add a vertex $s_{v}$. Call these two sets of vertices $E$ and $V$ respectively. Now add an edge from $s_{uv}$ to both $s_{u}$ and $s_{v}$ for each $(u,v) \in E(G)$. Add a root vertex $s$, and add vertices $a_{i,j}$ for all $1 \leq i \leq k-1$ and $1 \leq j \leq k$. Connect $a_{i,j}$ to $a_{i',j'}$ ($i' = i+1$) for all $i,j,j'$, connect $a_{1,j}$ to $s$ for all $j$, and connect $a_{k-1,j}$ to each vertex of $V$ for all $j$. Now set $k' = k + {k \choose 2} + 1$. 

We claim that \save on $(G',s,k')$ is a \textsc{Yes}-instance if and only if \textsc{$k$-Clique} on $(G,k)$ is a \textsc{Yes}-instance. Suppose that $G$ has a $k$-clique $K$. Then the strategy that protects the vertices $s_{v}$ for all $v \in K$ saves the vertices $s_{uv}$ for all $u,v \in K$. Since $K$ is a clique, these vertices $s_{uv}$ are indeed present in $G'$. Additionally, we can protect (and thus save) a vertex $s_{xy}$ for some edge $xy \not\in E(G[K])$. This edge exists, as $G$ is assumed to have at least $k+1$ nonisolated vertices. It follows that this strategy saves at least $k'$ vertices.

Suppose that $P= \{p_{1},\ldots,p_{\ell}\}$ is a strategy for $(G',s,k')$ that chooses vertex $p_{t}$ at time $t$ and saves at least $k'$ vertices. First observe that if $p_{t} = a_{i,j}$ for some $i,j$, then this vertex is not helpful, as there is always a vertex $a_{i,j'}$ that will be burned at time $t$ and has the same neighborhood as $a_{i,j}$. Hence we can assume that no vertex $a_{i,j}$ is protected by the strategy. This implies that all vertices of $V$ will be burned, except those that are protected by the strategy. But then protecting vertices of $E$ does not save any further vertices. Since the fire will reach $V$ in $k$ time steps, and thus $E$ in $k+1$ time steps, the vertices in $S \cap V$ are responsible for saving ${k \choose 2}$ vertices, which is only possible if the vertices of $S \cap V$ induce a $k$-clique in $G$.
\end{proof}
Observe that essentially the same construction works for the decision variant of \maxprotect.

\begin{figure}
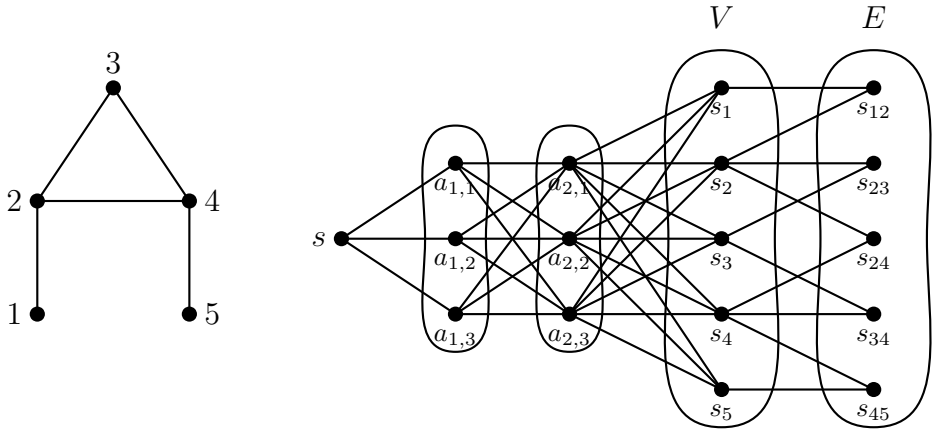

\begin{center}
\ig{saving-hardness.mps}
\end{center}
\caption{An instance of \textsc{$k$-Clique} and the corresponding graph $G'$ constructed in the proof of Theorem~\ref{thm:saving:hardness} for $k=3$.}
\label{fig:saving:hardness}
\end{figure}
\medskip

\begin{theorem}
\maxprotectdecision is W[1]-hard, even on bipartite graphs.
\end{theorem}
\begin{proof}
We again reduce from \textsc{$k$-Clique} and construct the same bipartite graph as in the proof of Theorem~\ref{thm:saving:hardness}. We set $k' = k+1$ and $K' = k + {k \choose 2} + 1$. Correctness now follows straightforwardly from the arguments in the proof of Theorem~\ref{thm:saving:hardness}.
\end{proof}
The above reduction also implies an NP-hardness reduction, which is simpler than the original reduction for the \firefighter problem on bipartite graphs~\cite{MacGillivrayW03}.

\subsection{Improved Algorithm on Trees}
We show that \save and \maxprotect have a deterministic $O(n + 2^{k}k^{3})$ and $O(2^{k}kn)$ algorithm, respectively, on trees. This resolves an open question of Cai \etal\cite{cai-verbin-yang}. As a consequence, we also obtain a refined subexponential algorithm for the \firefighter problem on trees, running in time $O(2^{\sqrt{2n}} n^{3/2})$.

The following observation is by MacGillivray and Wang \cite[Sect.~4.1]{MacGillivrayW03}.
\begin{lemma} \label{lem:trees:base}
For any optimum strategy for an instance of the \firefighter problem on trees, there is an integer $\ell$ such that all protected vertices have depth at most $\ell$, exactly one vertex $p_{i}$ at each depth $1 \leq i \leq \ell$ is protected, and all ancestors of each $p_{i}$ are burned.
\end{lemma}
We need the following notation. Let $T$ be any rooted tree. Use a pre-order traversal of $T$ to number the vertices of $T$ from $1$ to $n$. We say that $u \in V(T)$ is \emph{to the left} of $v \in V(T)$ if the number assigned to $u$ is not greater than the number of $v$ in the order. It is then easy to define what the \emph{leftmost} or \emph{rightmost} vertex is.%, or the \emph{rightmost vertex to the left} of a vertex is. %Finally, we say that $u$ is \emph{below} $v \not= u$ if $v$ is on the path in $T$ from $u$ to the root of $T$.

\begin{theorem} \label{thm:trees:vertex-protection}
\maxprotect has an $O(2^{k}kn)$-time algorithm on trees.
\end{theorem}
\begin{proof}
Let $(T,s,k)$ be an instance of \maxprotect on a tree $T$. Assume that $T$ is rooted at $s$. By Lemma~\ref{lem:trees:base}, we can define a characteristic vector $\chi_{v}$ of length $k$ for each vertex $v$ of the tree, which has a $1$ at position $i$ if and only if the optimal strategy protects a vertex at depth $i$ in the part of the tree to the left of $v$. We use these vectors as the basis for a dynamic programming procedure. However, the vector cannot ensure that no ancestors of a protected vertex will be protected. To ensure this, we add another dimension to our dynamic programming procedure. The pre-order numbering ensures that no descendant is protected.

The dynamic programming algorithm is then as follows.
Let $L$ be the set of vertices in $T$ that are at depth at most $k$. 
For each $v \in L$, let $P_{v}$ denote the path in $T$ between $v$ and $s$.
For each vector $\chi \subseteq \{0,1\}^{k}$ and each integer $0 \le i \le k$,  we compute $A_{v}(\chi,i)$, the maximum number 
of vertices one can save when protecting at most one vertex at depth $j$ for each $j$ for which $\chi(j)=1$ and no vertex otherwise,
where protected vertices must lie to the left of $v$ but at depth greater than $i$ when lying on $P_{v}$, and no protected vertex is an ancestor of another.
%but at most one vertex of $P_{v}$ at depth greater than $i$.
%and additionally no vertex from the path $P_{v}$ of depth at most $i$ is protected.
Observe that $s$ is the leftmost vertex of $L$. Now set $A_{s}(\chi,i)=0$ for any $\chi$ and $i$. Then
\begin{eqnarray*}
\label{eq:saving}
A_{v}(\chi,i) & = \max \big\{ & A_{l(v)}(\chi,\min\{\depth(v)-1,i\}),\\
  & & [\chi(\depth(v))=1 \wedge \depth(v)>i] \cdot \\
  & & \ (r(v) + A_{l(v)}(\chi^{v},\depth(v)-1)) \ \big\}
\end{eqnarray*}
Here $\depth(v)$ is the depth of a vertex $v$,
$l(v)$ is the rightmost vertex in $L$ which has strictly smaller value in the pre-order than $v$, and
$r(x)$ is the number of vertices saved when protecting only $x$.
Moreover, $\chi^{v}$ is the $0$-$1$ vector obtained from $\chi$ by setting the number at the index of $\chi$ corresponding to $\depth(v)$ to $0$.
In the formula we use Iverson's bracket notation, where $[\phi]$ is equal to one
if $\phi$ is true and zero otherwise.

To see that the above formula is correct, observe that we can either protect the considered vertex $v$ or not.
If we do not protect $v$, then we must ensure that the value for the second dimension of our dynamic programming procedure does not exceed the length of $P_{v}$, yet still captures the same forbidden part of $P_{v}$. Correctness then follows from the fact that the parent of $v$ is always on $P_{l(v)}$.
If we do protect $v$, we can protect $v$ only if we are allowed to do so, \ie if $\chi(\depth(v))=1$ and $\depth(v)>i$.
Furthermore, we need to ensure that no ancestor of $v$ is protected later. Therefore, we set the value for the second dimension of our dynamic programming procedure to $\depth(v)-1$.

To get the solution for the whole tree $T$, return $A_{v^{*}}(1^{k},0)$, where $v^{*}$ is the rightmost vertex of $L$.
To obtain the claimed running time, first find $L$, and then $l(v)$ for each vertex $v \in L$.
This can be done in linear time by a depth-first search. 
We can also compute $r(x)$ for each $x \in V(T)$ in linear time, 
as $r(x)$ equals one plus the number of descendants of $x$. 
By traversing the vertices of $L$ from left to right, the total running time is $O(2^{k}kn)$.
\end{proof}

\begin{corollary} \label{cor:trees:saving}
\save has an $O(2^{k}kn)$-time algorithm on trees.
\end{corollary}
\begin{proof}
Let $(T,s,k)$ be an instance of \save on a tree $T$. We run the above algorithm for \maxprotect for all $k' = 1,\ldots,k$. Observe that it is possible to save $k$ vertices of the tree if and only if the algorithm saves at least $k$ vertices for some value of $k'$. Furthermore, we note that
$$\sum_{i=1}^{k} (2^{i} i n)\ \leq\ kn \sum_{i=1}^{k} 2^{i} \ =\ (2^{k+1}-2)kn,$$
implying that the worst-case running time of the algorithm is $O(2^{k}kn)$.
\end{proof}
Using the known kernel of size $O(k^{2})$~\cite{cai-verbin-yang}, we can improve running time of the above algorithm for \save to $O(n + 2^{k}k^{3})$.

To obtain a good subexponential algorithm, we use the following lemma.

\begin{lemma} \label{lem:trees:sub-aux}
If a vertex at depth $d$ burns in an optimum strategy for an instance of the \firefighter problem on trees, then at least $\frac{1}{2}(d^{2} + d)$ vertices are saved.
\end{lemma}
\begin{proof}
Let $(T,s)$ be an instance of the \firefighter problem on trees, and let $v$ be a vertex of depth $d$ that burns in an optimum strategy. Then the strategy protects a vertex at depth $d$, and by Lemma~\ref{lem:trees:base} it thus protects a vertex $p_{i}$ at each depth $i$ for $1 \leq i \leq d$. For any $i$, the subtree rooted at $p_i$ should contain at least $d-i+1$ vertices, or it would have been better to protect the vertex at depth $i$ that is on the path from $v$ to $s$. But then the strategy saves at least $\sum_{i=1}^{d} (d-i+1) = \frac{1}{2}(d^{2} + d)$ vertices.
\end{proof}

\begin{theorem} \label{thm:trees:sub}
The \firefighter problem has an $O(2^{\sqrt{2n}} n^{3/2})$-time algorithm on trees.
\end{theorem}
\begin{proof}
Let $(T,s)$ be an instance of the \firefighter problem on trees. Suppose that a vertex $v$ at depth $\sqrt{2n}$ burns in an optimum strategy. Then, by Lemma~\ref{lem:trees:sub-aux}, the strategy saves at least $n + \sqrt{n/2} > n$ vertices, which is not possible. It follows that all vertices at depth $\sqrt{2n}$ are saved in any optimum strategy. Since in any optimum strategy every protected vertex has a burned ancestor by Lemma~\ref{lem:trees:base}, all protected vertices are at depth at most $\sqrt{2n}$. Hence there is an optimum strategy that protects at most $\sqrt{2n}$ vertices, and we can find the optimum strategy by running the algorithm of Theorem~\ref{thm:trees:vertex-protection} with $k = \sqrt{2n}$.
\end{proof}

\subsection{Tractability on Graphs of Bounded Treewidth}
We generalize the above results by showing that \maxprotect and \save remain fixed-parameter tractable when parameterized by $k$ and the treewidth of the underlying graph. To this end, we use Monadic Second Order Logic (MSOL).
%The proof of the following theorem can be found in Appendix~\ref{app:treewidth}.
The syntax of MSOL of graphs includes the logical connectives $\vee$, $\land$, $\neg$, $\Leftrightarrow $,  $\Rightarrow$, variables for vertices, edges, sets of vertices, and sets of edges, the quantifiers $\forall$, $\exists$ that can be applied to these variables, and the following four binary relations: 
\begin{enumerate}
\item $u\in U$, where $u$ is a vertex variable and $U$ is a vertex set variable.
\item $d \in D$, where $d$ is an edge variable and $D$ is an edge set variable.
\item ${\adj}(u,v)$, where $u, v$ are vertex variables, and the interpretation is that $u$ and $v$ are adjacent.
\item Equality, $=$, of variables representing vertices, edges, sets of vertices, and sets of edges.
\end{enumerate}
For \maxprotect, we actually need Linear Extended MSOL~\cite{ArnborgLS1991}, which allows the maximization over a linear combination of the size of unbound set variables in the MSOL formula. (The definition of LEMSOL in~\cite{ArnborgLS1991} is slightly more general, but this suffices for our purposes.)

\begin{theorem}
\label{thm:saving-treewidth}
\maxprotect is in FPT when parameterized by $k$ and the treewidth of the graph.
\end{theorem}
\begin{proof}
Let $(G,s,k)$ be an instance of \maxprotect such that the treewidth of $G$ is $t$. Use Bodlaender's Algorithm~\cite{Bodlaender1996} to find a tree decomposition of $G$ of width at most $t$. Consider the following MSOL formulae.
\begin{eqnarray*}
\lefteqn{\mathrm{NextBurn}(B_{i-1}, B_{i}, p_{1},\ldots,p_{i})\ :=} \\
&\forall v\, \Big(\Big(v \in B_{i-1} \vee \exists u\, \Big(u \in B_{i-1} \wedge \adj(u,v) \wedge \Big(\bigwedge_{1 \leq j \leq i} v \not= p_{j}\Big)\Big)\Big) \Leftrightarrow v \in B_{i}\Big)
\end{eqnarray*}
This expresses is that if the vertices of $B_{i-1}$ are burning by time step $i-1$, then the vertices of $B_{i}$ burn by time step $i$, assuming that vertices $p_{1},\ldots,p_{i}$ have been protected so far.
\begin{eqnarray*}
\lefteqn{\mathrm{Saved}(S,B, p_{1},\ldots,p_{\ell})\ :=}\\
&& \forall u \Big( u \in S \Rightarrow \Big(u \not\in B \wedge \forall v \Big(\adj(u,v) \Rightarrow v \in S \vee \bigvee_{1 \leq i \leq \ell} p_{i} = u \Big) \Big) \Big)
\end{eqnarray*}
This expresses that $S$ is a set of saved vertices when $B$ is a set of burned vertices and vertices $p_{1},\ldots,p_{\ell}$ are protected.
\begin{eqnarray}
\lefteqn{\mathrm{Protect}(S, \ell) := \exists p_{1},\ldots,p_{\ell}\ \exists B, B_{0},\ldots,B_{\ell-1}}\hspace{2cm} \nonumber \\
&& \quad  \forall u\ (u \in B_{0} \Leftrightarrow u = s) \label{eq:1} \\
&& \wedge\ \bigwedge_{1 \leq i \leq \ell-1} \mathrm{NextBurn}(B_{i-1},B_{i},p_{1},\ldots,p_{i}) \label{eq:2} \\
&& \wedge\ \bigwedge_{1 \leq i \leq \ell} p_{i} \not\in B_{i-1} \label{eq:3} \\
&& \wedge\ \forall u\ \Big( \Big(\bigvee_{0 \leq i \leq \ell-1} u \in B_{i}\Big) \Rightarrow u \in B \Big) \label{eq:4} \\
&& \wedge\ \mathrm{Saved}(S,B,p_{1},\ldots,p_{\ell}) \label{eq:5}
\end{eqnarray}
This expresses that $S$ can be saved by protecting $\ell$ vertices. The sets $B_{i}$ contain all vertices that are burned by time step $i$, which is ensured by the formulas in lines \ref{eq:1} and \ref{eq:2}. The set $B$ contains vertices that are not saved (line \ref{eq:5}) and all vertices of the sets $B_{i}$ (line \ref{eq:4}). The vertices $p_{1},\ldots,p_{\ell}$ are the vertices that are protected. Line \ref{eq:3} ensures that the vertices we want to protect are not burned by the time we pick them. Then we want to find the largest set $S$ such that
$$\mathrm{Protect}_{k}(S) := \bigvee_{1 \leq \ell \leq k} \mathrm{Protect}(S,\ell)$$
is true. Following a result of Arnborg, Lagergren, and Seese~\cite{ArnborgLS1991}, this can be done in $f(k,t) \cdot n^{O(1)}$ time using the above formula.
\end{proof}
In the same way as Corollary~\ref{cor:trees:saving}, we then obtain the following.

\begin{corollary}
\save is in FPT when parameterized by $k$ and the treewidth of the graph.
\end{corollary}
Observe that this algorithm also works on graphs of bounded local treewidth, because if the graph has a vertex at distance more than $k$ from the root, then any strategy that protects a vertex at distance $i$ from the root in time step $i$ will save at least $k$ vertices, and we can answer \textsc{Yes} immediately.

\begin{corollary}
\save is in FPT on graphs of bounded local treewidth.
\end{corollary}
The class of graphs having bounded local treewidth coincides with the class of apex-minor-free graphs~\cite{Eppstein2000}, which includes the class of planar graphs.

\begin{corollary}
\save is in FPT on planar graphs.
\end{corollary}

\section{Burning Vertices} \label{sec:burning}
In this section, we consider the \textsc{Firefighter} problem when parameterized by the number of burned vertices, which we call the \burn problem.
%In~\cite{cai-verbin-yang} Cai, Verbin and Yang show $O(4^kn)$ randomized and $2^{O(k)}n\log n$ deterministic algorithm solving \burn for trees.
We improve on results of Cai \etal\cite{cai-verbin-yang} by showing an $O(2^kn)$-time deterministic algorithm for trees and an $O(3^kn)$-time deterministic algorithm for general graphs.
Furthermore, we prove that the \burn problem has no polynomial kernel for trees, resolving an open problem from~\cite{cai-verbin-yang}.

%\subsection{Algorithm on Trees}
%\input{burning-algorithm-trees}

\subsection{Algorithms}
In this subsection, we show an $O(2^kn)$-time algorithm for the \burn problem on trees, and an $O(3^kn)$-time algorithm on general graphs.

\begin{theorem}
The \burn problem for trees can be solved in $O(2^kn)$ time and polynomial space.
\end{theorem}
\begin{proof}
If the root $s$ has at most one child, then we immediately answer \textsc{Yes}.
We may assume that the root has exactly $a \ge 2$
children, and $k \ge a-1$ since otherwise we simply answer \textsc{No}.
We use Lemma~\ref{lem:trees:base}
and branch on every child of the root $s$.
In each branch, we cut the subtree rooted at the protected vertex,
identify all the vertices that are on fire after the first round,
and decrease the parameter by $a-1$.
In this way, we obtain a new instance of the \burn problem
with parameter value equal to $k-(a-1)$.
The time bound follows from the inequality 
$$T(k) \le aT(k-(a-1))+O(n)$$
which is worst when $a=2$.
%\tudu{MC: It is obvious or should we be more elaborate here?}
\end{proof}
Before we present the algorithm on general graphs, we need to reformulate the \firefighter problem
to an equivalent version. Consider a different version of the \firefighter problem,
where in each round an arbitrary number of vertices may be protected under the following restrictions:
\begin{itemize}
  \item each protected vertex must have a neighbor which is on fire,
  \item after $i$ rounds of the process at most $i$ vertices are protected.
\end{itemize}
By \burntwo we denote the \burn problem where vertices are protected subject to the above rules.

\begin{lemma}
An instance $(G,s,k)$ of the \burn problem is a \textsc{Yes}-instance if and only if it is a \textsc{Yes}-instance of the \burntwo problem.
\end{lemma}

\begin{proof}
Assume that $(G,s,k)$ is a \textsc{Yes}-instance of the \burn problem.
Let $P$ be the set of protected vertices of an optimum strategy $S$. 
We construct a strategy $S'$, which in the $i$-th round of \burntwo protects 
exactly those vertices of $P$ which have a neighbor which is on fire.
Clearly after $i$ rounds at most $i$ vertices will be protected, since 
each vertex of $P$ is protected by the strategy $S'$ not earlier than by the strategy $S$.

In the other direction assume that $(G,s,k)$ is a \textsc{Yes}-instance of the \burntwo problem
and $P$ is the set of protected vertices of an optimum strategy $S'$.
We construct a strategy $S$ as follows.
Let $(v_1,\ldots,v_{|P|})$ be a sequence of vertices of $P$ sorted by the
round in which a vertex is protected by $S'$ (breaking ties arbitrarily).
In the $i$-th round of strategy $S$ we protect the vertex $v_i$.
The vertex $v_i$ is not on fire in the $i$-th round, because in the strategy $S'$ 
it is protected not earlier than in the $i$-th round.
\end{proof}

\begin{theorem}
There is an $O(3^kn)$-time and polynomial-space algorithm for the \burntwo problem on general graphs.
\end{theorem}

\begin{proof}
We present a simple branching algorithm. Assume that we are in the $i$-th time step
and let $B$ be the set of vertices which are currently on fire.
Moreover, let $P$ be the set of already protected vertices (in the first round we have $B=\{s\}$ and $P=\emptyset$).
Let $a=i-|P|$ and $r=|N(B) \setminus P|$.
The algorithm does the following:
\begin{enumerate}
   \item \label{step1} If $|B| > k$, then we immediately answer \textsc{No}.
  \item \label{step2} Observe that in the $i$-th round we are allowed to protect at most $\min(a,r)$ vertices.
If $a \ge r$, then we can greedily protect the whole set $N(B) \setminus P$. Hence in this case we answer \textsc{Yes}.
  \item In the last case, when $a < r$, we branch on all subsets of $N(B) \setminus P$ of size at most $a$.
Observe that the number of branches is equal to $\sum_{j=0}^{a} \binom{r}{j} \le 2^r-1$, since we have $a < r$.
\end{enumerate}
The running time of the algorithm is as follows. 
We introduce a measure $\alpha=(k-|B|)+(i-|P|)$ which we use in our time bound.
At the beginning of the first round of the burning process we have $\alpha=(k-1)+(1-0)=k$.
By $T(\alpha)$ we denote the upper bound on the number of steps that our algorithm requires 
for a graph with measure value $\alpha$.
Observe that for $\alpha \le 0$, we have $T(\alpha) = O(n)$.
Let us assume that the algorithm did not stop in step~\ref{step1} or \ref{step2}, and
it branches into at most $2^r-1$ choices of protected vertices.
Observe that no matter how many vertices the algorithm decides to protect,
the value of the measure decreases by exactly $r-1$.
Consequently, we have the inequality $T(\alpha) \le (2^r-1)T(\alpha-r+1)+O(n)$.
Since the algorithm did not stop in steps~\ref{step1} or \ref{step2}, we infer that $r \ge 2$.
The time bound follows from the fact that the worst case for the inequality occurs when $r=2$.
\end{proof}

\begin{corollary}
There is an $O(3^kn)$-time and polynomial-space algorithm for the \burn problem on general graphs.
\end{corollary}

\subsection{No Poly-Kernel for Trees}
\newcommand{\compass}{\ensuremath{\textrm{NP} \subseteq \textrm{coNP}/\textrm{poly}}}

The aim of this subsection is to prove the following theorem.

\begin{theorem}
\label{thm:burn-nokernel}
Unless $\compass$, there is no polynomial kernel for the \burn
problem, even if the input graph is a tree of maximum degree four.
\end{theorem}
Before we prove Theorem~\ref{thm:burn-nokernel} we describe
the necessary tools.
We use the cross-composition technique introduced by Bodlaender et al. \cite{cross-composition}, which is based on the previous results of Bodlaender et al.~\cite{bodlaender-nokernel} and Fortnow and Santhanam~\cite{fortnow-santhanam-nokernel}.
We recall the crucial definitions.
\begin{definition}[Polynomial equivalence relation \cite{cross-composition}]\ 
An equivalence relation $\mathcal{R}$ on $\Sigma^\ast$ is called a {\em{polynomial equivalence
  relation}} if (1) there is an algorithm that given two strings $x,y \in \Sigma^\ast$
  decides whether $\mathcal{R}(x,y)$ in $(|x|+|y|)^{O(1)}$ time; (2) for any finite set $S \subseteq \Sigma^\ast$
  the equivalence relation $\mathcal{R}$ partitions the elements of $S$ into at most $(\max_{x \in S} |x|)^{O(1)}$ classes.
\end{definition}
\begin{definition}[Cross-composition \cite{cross-composition}]
Let $L \subseteq \Sigma^\ast$, and let $Q \subseteq \Sigma^\ast \times \mathbb{N}$ be
a parameterized problem. We say that $L$ {\em{cross-composes}} into $Q$ if there is a polynomial
equivalence relation $\mathcal{R}$ and an algorithm which, given $t$ strings $x_1, x_2, \ldots x_t$
belonging to the same equivalence class of $\mathcal{R}$, computes an instance
$(x^\ast,k^\ast) \in \Sigma^\ast \times \mathbb{N}$ in time polynomial in $\sum_{i=1}^t |x_i|$
such that (1) $(x^\ast,k^\ast) \in Q$ iff $x_i \in L$ for some $1 \leq i \leq t$; (2)
$k^\ast$ is bounded polynomially in $\max_{i=1}^t |x_i| + \log t$.
\end{definition}
\begin{theorem}[\cite{cross-composition}, Theorem 9]\label{thm:cross-composition}
If $L \subseteq \Sigma^\ast$ is NP-hard under Karp reductions
and $L$ cross-composes into the parameterized problem $Q$ that
has a polynomial kernel, then $\compass$.
\end{theorem}
We apply Theorem \ref{thm:cross-composition}, where 
as the language $L$ we use \burn in trees of maximum degree three, which is
NP-complete~\cite{FinbowKMR07}.
%
%\begin{theorem}[TODO ref]
%The unparameterized version of the \burn problem is NP-complete 
%even on trees with maximum degree three.
%\end{theorem}
%
To finish the proof of Theorem~\ref{thm:burn-nokernel}, we present a cross-composition algorithm. 
%This is done in the following lemma.

\begin{lemma}
\label{lem:burn-cross-composition}
The unparameterized version of the \burn problem 
on trees with maximum degree three
cross-composes to \burn on trees with maximum degree four.
\end{lemma}

\begin{proof}
Observe that any polynomial equivalence relation is defined on all
words over the alphabet $\Sigma$ and for this reason
we should also define how the relation behaves on words that do not represent instances of the problem.
For the equivalence relation $\mathcal{R}$ we take a relation
that puts all malformed instances into one equivalence class and all
well-formed instances are grouped according to the number of vertices we
are allowed to burn. 

If we are given malformed instances, we simply output a trivial \textsc{No}-instance.
Thus in the rest of the proof we assume we are given
a sequence of instances $(T_i,s_i,k)_{i=1}^t$ of
the \firefighter problem, where each $T_i$ is of maximum degree three.
Observe that in all instances we have the same value of the parameter $k$.
W.l.o.g.\ we assume that $t=2^h$ for some integer $h \ge 1$. Otherwise we can duplicate an appropriate number of instances $(T_i,s_i,k)$.

We create a new tree $T'$ as follows. Let $T'$ be a full binary
tree with exactly $t$ leaves rooted at a vertex $s'$.
Now for each $i=1,\ldots,t$, we replace
the $i$-th leaf of the tree by tree $T_i$ rooted at $s_i$.
Finally, we set $k'=k+h = k + \log_2t$.
%\tudu{MC: Do we need a figure here? I think that it is so easy that it is not needed.}
Observe that since each tree $T_i$ is of maximum degree three, the tree $T'$
is of maximum degree four.
To prove correctness, it is
enough to show that any strategy that minimizes 
the number of burned vertices protects exactly one vertex at each depth $1,\ldots,h$, which follows from Lemma~\ref{lem:trees:base}.
%since if there exists a depth $1 \le d \le h$ where the strategy does not protect any
%vertex than we may take a vertex $v_0$ with the smallest depth greater than $d$ and
%replace it by its ancestor at depth $d$ which decreases the number of burned vertices
%and contradicts the optimality of the strategy.
Hence in any strategy that minimizes the number of burned vertices, there will be exactly one vertex $s_i$ which is on fire after $h$ rounds.
\end{proof}
We can obtain a similar result for the decision variant of \maxprotect.
%Due to space limitations the proof can be found in Appendix~\ref{app:kernel}.

\begin{theorem}
\label{thm:protection-kernel}
Unless $\compass$, there is no polynomial kernel for the \maxprotectdecision
problem, even if the input graph is a tree of maximum degree four.
\end{theorem}
\begin{proof}
There are only two differences compared to the proof for \burn.
\begin{itemize}
  \item For the equivalence relation $\mathcal{R}$, we take a relation
that puts all malformed instances into one equivalence class, and all
well-formed instances are grouped according to the number of vertices 
of the tree, the parameter value $k$, and the value $K$.
  \item The value of $k'$ for the tree $T'$ is $k+h$, and the value of $K'$ is equal to $K+(t-1)n+(t-h-1)$,
  where $n$ is the number of vertices in each of the trees $T_i$.
  The additional summands are derived from the fact that any optimal strategy
  will ensure that after $h$ rounds exactly one vertex $s_i$ will be on fire
  and hence we save $t-1$ subtrees rooted at $s_i$,
  each containing $n$ vertices, and $t-h-1$ vertices of the full binary tree.
\end{itemize}
%\tudu{MC: If you think that full proof is needed please email me and I will type it.}
This completes the proof.
\end{proof}

%\section{Nonstandard Parameters}
%\input{nonstandard}

\section{Open Problems} \label{sec:open}
In this paper, we refined and extended several parameterized algorithmic and complexity results about different parameterizations of the \firefighter problem.
We conclude with the following open problems.  
\begin{itemize}
\item We have shown that \save is in FPT on graphs of bounded local treewidth , and thus on planar graphs, by showing that the problem is in FPT parameterized by $k$ and the treewidth of a graph. While 
 \maxprotect  is also in FPT parameterized by $k$ and the treewidth, we do not know if the problem is in  
 FPT on planar graphs, and leave it as an open problem. 
 \item  The \firefighter problem is solvable in subexponential time on trees. 
Is it   solvable  in time $2^{o(n)}$ on $n$-vertex  planar graphs? Even the case of outerplanar graphs is open. 
\item   
Finally, we do not know  if any of the three parameterized versions  of the problem is solvable in parameterized subexponential time $2^{o(k)}n^{O(1)}$ on trees. 
\end{itemize}

\paragraph{Acknowledgement}
We thank Leizhen Cai for pointing us to~\cite{Yang2009} and for sending us the full version of~\cite{cai-verbin-yang}. We also acknowledge the support of Schloss Dagstuhl for Seminar 11071 (GRASTA 2011 - Theory and Applications of Graph Searching Problems).  Research of Fedor Fomin  was supported by the European Research Council (ERC) grant ÒRigorous Theory of PreprocessingÓ, reference 267959. 

\bibliographystyle{splncs03}
\bibliography{index}

%\newpage

%\appendix

%\input{appendix}

\end{document}